\documentclass[conference,a4paper,10pt]{IEEEtran}
\usepackage{multirow}
\usepackage{subfigure}
\usepackage{graphicx}
\usepackage{longtable}
\usepackage{courier}
\usepackage{float}
\usepackage{cite}
\usepackage{amsmath,amssymb}
\usepackage{hyperref}
\usepackage{bbold}
\usepackage{dsfont}
\usepackage[in]{fullpage}
\usepackage[utf8]{inputenc}
\usepackage[english]{babel}
\usepackage{amsthm}
\usepackage{xcolor}

\newtheorem{theorem}{Theorem}

\newcommand{\dfn}{\stackrel{\triangle}{=}}
\newcounter{MYtempeqncnt}

\usepackage[T1]{fontenc}
\usepackage[utf8]{inputenc}
\usepackage{authblk}
\usepackage{array}
\usepackage{url}

\usepackage[margin=0.53in]{geometry}

\newtheorem{definition}{Definition}[section]

\begin{document}
\title{Achievability Bounds of Coding with Finite Blocklength for Gaussian Broadcast Channels 
}
\author{%
Ay\c{s}e \"{U}nsal\\ \url{unsal@eurecom.fr}
 \and 
Jean-Marie Gorce\\ \url{jean-marie.gorce@insa-lyon.fr}
}

\maketitle
 \thispagestyle{plain}
\begin{abstract}
In this paper, we study the achievable performance of dirty paper coding for the Gaussian broadcast channel (BC) with finite blocklength and we propose two different achievability bounds for this problem. We present the broadcast adaptation of dependence testing bound of Polyanskiy \textit{et al.} \cite{Polyanskiy}, which is an upper bound on the average error probability that depends on the \textit{channel dispersion terms} of each error event for fixed input. Additionally, we introduce the $\kappa \beta$ lower bounds on the maximal code sizes of each user using dirty paper coding\footnote{This paper was presented at the IEEE ITW 2017 as poster.}.
\end{abstract}

\section{Introduction}
Classical information theory focuses on understanding and determining fundamental limits of communication systems using coding schemes with asymptotically long blocklengths. However in practice, the requirement for coding with finite blocklength resulted in search for a measure of loss in the performance of the systems due to coding with finite blocklength. In \cite{Strassen}, the author showed that, the maximum code size $M^{*}$ of a discrete-memoryless point-to-point channel equals to
\begin{equation}\label{eq:poli}
\log M^{*}(n,\epsilon)=nC-Q^{-1}(\epsilon)\sqrt{nV}+O(\log n)
\end{equation} where the blocklength, the error probability and the channel capacity are denoted by $n$, $\epsilon$ and $C$, respectively. The loss in the capacity, which is a first order statistic, is measured by a second order statistics called the channel dispersion denoted by $V$ in (\ref{eq:poli}). 

It is well known that dirty paper coding achieves the optimal performance of superposition coding and successive cancellation for the Gaussian BC in the asymptotic regime \cite{GelfandPinsker}. 
This paper focuses on the non-asymptotic performance of the dirty paper coding for Gaussian BC and provides an upper bound on the average error probability through deriving the channel dispersion for each error event that is defined.
To this end, we adapt two achievability bounds introduced
by Polyanskiy et al. for a single-user channel in \cite{Polyanskiy} to the broadcast setting. Namely, we introduce the dependence testing (upper) bound on the average error probability of the Gaussian BC through deriving the channel dispersion terms for each error event that is defined in addition to the $\kappa \beta$ lower bounds on the maximum code sizes of each user using dirty paper coding and threshold decoding.

In \cite{Polyanskiy}, (\ref{eq:poli}) is proved to hold for the single user AWGN channel where the channel dispersion term is derived in a closed form, as a function of the signal-to-noise ratio. The dispersion of dirty paper coding is analyzed in \cite{Scarlett} for a given error probability in a single user channel. The author of \cite{Scarlett} showed that, the dispersion of Costa's dirty paper coding \cite{Costa} does not change in the absence of the channel state.
In \cite{Laneman_3}, each component of the channel dispersion matrix is derived as a function of the signal-to-noise ratio. The main contribution of \cite{Laneman_3} is the maximum achievable rate region of the Gaussian MAC. In a recent work \cite{UG_2017}, the authors presented the dependence testing and $\kappa \beta$ bounds for the two-receiver Gaussian BC which are originally proposed by Polyanskiy \textit{et al.} in \cite{Polyanskiy} for the single user AWGN channel. 

In the next part, we give a detailed description of the considered system model. Section \ref{sec:related} reminds the reader some results from the literature regarding the asymptotic and non-asymptotic performance of different coding techniques for the Gaussian BC. In Sections \ref{sec:DT_DPC} and \ref{sec:kb_DPC}, we respectively present the adaptations of dependence testing and $\kappa \beta$ bounds for the two-receiver Gaussian BC in the setting of dirty paper coding and threshold decoding. 

\section{System Model \label{sec:model}}
In order to communicate the messages $m_1 \in [1,2,\cdots,M_1]$ and $m_2 \in [1,2,\cdots,M_2]$ to two users, we define the following real-valued two-receiver Gaussian broadcast channel 
\begin{equation} \label{eq:degraded_GBC_channel}
Y_{j,i}=X_{i}+Z_{j,i},
\end{equation} with the channel input 
$X_{i}$, the output signals $Y_{j,i}$ and the channel noise terms $Z_{j,i} \sim \mathcal{N}(0,N_j)$ 
for $j=1,2$ and $i=1,\cdots,n$. The input signal is 
\begin{equation}
X^n=X_1^n+X_2^n \label{eq:DPC_input}
\end{equation} where $X_1^n$ and $X_2^n$ are independent signals.
The channel transition probability density is given as
\begin{align}
P_{Y_j^n|X^n}(y_j^n|x^n)&= \Pi_{i=1}^{n} P_{Y_j|X^n}(y^n_j|x) \nonumber \\
&=(2\pi N_j)^{-n/2} \mathrm{e}^{-\frac{\|y_j^n-x^n\|^2}{2N_j}}.
\end{align}
According to the multiuser adaptation of the dirty paper coding in the asymptotic regime introduced by Gelfand and Pinsker, 
to communicate $m_2$ to user 2 with input signal $X_2^n$, interfering signal $X_1^n$ and the noise term $Z_2^n$, the interfering signal is treated as noise. To transmit $m_1$ to user 1 with the input signal $X_1^n$, Gaussian noise term $Z_1^n$ and $X_2^n$ is treated as the channel state which is known at the encoder. 

The information density $i(.;.)$ is defined by the following logarithmic ratio 
\begin{equation} \label{eq:inf_density}
i(x^n;y^n) \dfn
\log \frac{dP_{Y^n|X^n}(y^n|x^n)}{dP_{Y^n}(y^n)} 
\end{equation} where $P(.|.)$ and $P(.)$ respectively denote conditional and marginal probability distributions. 
The decoding rule is defined using threshold decoding as
\begin{equation} \label{eq:decod_rule}
i(x^n;y^n)>\gamma
\end{equation} where $\gamma$ denotes some threshold to be defined later.

\subsection{Error probability}
We define the overall average error probability based on either of the messages being decoded erroneously as follows
\begin{equation} \label{eq:er_prob_DP}
\epsilon \dfn \Pr\left[\{\hat{m}_1 \neq m_1\} \cup \{\hat{m}_2 \neq m_2\}\right]
\end{equation} where $\hat{m}_j $ denotes the estimated message $m_j$ for $j=1,2$.

\subsection{Random codebooks \label{subsec:model_DT}}

The achievability bounds presented in this paper are distinguished based on either fixing the input distribution as the dependence testing bound or fixing the output distribution as the $\kappa \beta$ bounds. 
Here we have an input and two output signals following the $n$-letter channel model given by (\ref{eq:degraded_GBC_channel}) and the input signal defined as in (\ref{eq:DPC_input}).
\begin{definition} \label{eq:def_code_GBC}
An $(n,M_1,M_2,\epsilon,P)$ code consists of the messages $m_1 \in [1,2,\cdots,M_1]$ and $m_2 \in [1,2,\cdots,M_2]$ that are respectively encoded into $n$ i.i.d. sequences $X_1 \sim N(0,\alpha P)$ and $X_2 \sim N(0,\bar{\alpha}P)$  that are chosen randomly 
and follow the average power constraints
\begin{equation} \label{eq:power_BC}
\frac{1}{M_1}\sum_{i=1}^{M_1} \mathbb{E}\|x_{1,i}\|^2  \leq n\alpha P,\;\;\;
\frac{1}{M_2}\sum_{i=1}^{M_2} \mathbb{E}\|x_{2,i}\|^2 \leq n\bar{\alpha} P
\end{equation} and the average probability of error defined by (\ref{eq:er_prob_DP}) where $\alpha $ is a constant confined in $[0,1]$ for $\bar{\alpha}=1-\alpha$ .
\end{definition} Definition \ref{eq:def_code_GBC} applies to the case of fixed input distribution and will be used in the derivation of the dependence testing bound on the error probability in Section \ref{sec:DT_DPC}.
\section{Related Work \label{sec:related}}
\subsection{Asymptotic Capacity}
In this part, we remind the reader of the asymptotic capacity region of the Gaussian BC in the setting of superposition coding (SPC) and successive cancellation. \cite{Bergmans, Gallager, Cover} showed that following inequalities hold for such a system as defined by (\ref{eq:degraded_GBC_channel}) with two receivers
\begin{align}
R_1& \leq C\left(\frac{\alpha P}{N_1}\right),\label{eq:R1} \\
R_2& \leq C\left(\frac{\bar{\alpha}P}{\alpha P +N_2}\right) \label{eq:R2}
\end{align} where $R_j =\frac{1}{n}\log_2 M_j$ and $N_2>N_1$. The Shannon capacity $C(.)$ is defined for a signal-to-noise ratio of $x$ as
${C(x)=\frac{1}{2} \log (1+x)}$. In \cite{GelfandPinsker}, Gelfand and Pinsker showed that the inequalities (\ref{eq:R1}) and (\ref{eq:R2}) are achieved by dirty paper coding even for the case of $N_1>N_2$. 
Namely, (\ref{eq:R1}) and (\ref{eq:R2}) are achieved respectively through the mutual information functions $I(U^n_1;Y^n_1)-I(U^n_1;U^n_2)$ and $I(U^n_2;Y^n_2)$ with the input signal defined as 
\begin{equation} 
X^n=(U_1^n-b_2 U_2^n)+U_2^n \label{DPC_inputb}
\end{equation} 
with $U^n_1=b_2 X^n_2+X^n_1$, $U^n_2=X^n_2$ where $b_2= \alpha P/(\alpha P+N_1)$.
\subsection{Non-asmyptotic performance of SPC}
The non-vanishing error probability of the Gaussian BC in the setting of SPC and successive cancellation as coding with finite blocklength is upper bounded in \cite{UG_2017} as given by the following theorem.
\begin{theorem}[\"{U}nsal and Gorce 2017 \cite{UG_2017}]
\label{theo:PEGBC}
An $(n,M_1,M_2,\epsilon,P)$ code exists for the channel as described in Definition \ref{eq:def_code_GBC} with the average error probability satisfying $\epsilon \leq \epsilon_{SP}$, where $\epsilon_{SP}$ is 
\begin{align} 
&\!\!\!\epsilon_{SP}= \Pr \left[ i(X_2^n;Y_2^n) < \gamma_2\right]+ \eta_2\Pr \left[i(X_2^n;\bar{Y}_2^n) > \gamma_2\right ] \label{eq:DT_SP_2} \\
&\!\!\!+ \Pr \left[i(X^n;Y_1^n|X_2^n)<\gamma_1\right]+\eta_1 \Pr \left [ i(X^n;\bar{Y}_1^n|X_2^n)> \gamma_1\right ]\label{eq:DT_SP_1} \\
&\!\!\!+\Pr \left[ i(X_2^n;Y_1^n) < \gamma_2\right]+ \eta_2\Pr \left[i(X_2^n;\bar{Y}_1^n) > \gamma_2\right ] \label{eq:DT_SP_21} \\
&\!\!\!+\Pr[X^n \notin F]\label{eq:ub_epsilon_2}
\end{align} 
for $\gamma_j=\log \eta_j$, $\eta_j = (M_j-1)/2$, $j=1,2$ and $\bar{Y}^n_j$ denotes a signal following the same distribution $Y^n_j$ that is independent of the input. $F$ denotes the set of permissible inputs due to the power constraint.
\end{theorem}
The next part focuses on the combination of Theorem \ref{theo:PEGBC} and \cite{GelfandPinsker} to upper bound the average error probability (\ref{eq:er_prob_DP}).

\section{Dependence Testing Bound --Writing on Dirty Paper \label{sec:DT_DPC}}
 
In the following theorem, we present our first main result.
\begin{theorem}\label{theo:DPC}
An $(n,M_1,M_2,\epsilon,P)$ code exists for the two receiver Gaussian BC with the average error probability (\ref{eq:er_prob_DP}) upper bounded by 
\begin{align} 
\!\!\!\epsilon &\leq  
\Pr \left[i(U_2^n;Y_2^n) < \gamma_2 \right] + \eta_2\Pr \left [  i(U_2^n;\bar{Y}_2^n)> \gamma_2\right ] \label{eq:user2_DPC}\\
&\!+\Pr \left[ i(U_1^n;Y_1^n) < \gamma'_1 \right]+\eta_1\Pr \left[i(U_1^n;\bar{Y}_1^n)  > \gamma'_1\right ] \label{eq:user1_DPC}\\
&\!+\Pr[X^n \notin F]\label{eq:ub_epsilon_22}
\end{align} where $\gamma'_1=\gamma_1+I(U_1^n;U_2^n)$, $\gamma_j=\log \eta_j$  and $\eta_j = (M_j-1)/2$ for $j=1,2$. $\bar{Y}^n_j$ follows the same distribution as the output $Y^n_j$ but is independent of the input $X^n$ defined by (\ref{eq:DPC_input}) and $F$ denotes the set of permissible inputs according to the input power constraint.
\end{theorem}
\begin{proof}
The average error probability (\ref{eq:er_prob_DP}) is bounded by the union bound as follows
\begin{align}\label{eq:er_prob_1}
&\Pr\left[\{ \hat{m}_1 \neq m_1\}\; \cup \; \{\hat{m}_2 \neq m_2\}\right]\leq \\
&\Pr\left[\{ \hat{m}_1 \neq m_1\}\right]+\Pr \left\{\hat{m}_2 \neq m_2\}\right]
\end{align} Let us denote $\Pr\left[\{ \hat{m}_j \neq m_j\}\right]$ by $\epsilon_j$ for $j=1,2$.
The capacity region presented in \cite{GelfandPinsker} is achieved by $I(U_1^n;Y_1^n) -I(U_1^n;U_2^n)$ for user 1 and $I(U_2^n;Y_2^n)$ for user 2. Applying the threshold decoding rule in (\ref{eq:decod_rule}) with the associated information density functions for each user, we get
\begin{align} 
i(U_2^n;Y_2^n) &> \gamma_2\label{eq:threshold_GBC_DPC2}  \\ 
i(U_1^n;Y_1^n) &>\gamma'_{1} \label{eq:threshold_GBC_DPC1} 
\end{align} 
The corresponding error probabilities of violating the conditions in (\ref{eq:threshold_GBC_DPC2}) and (\ref{eq:threshold_GBC_DPC1} ) are $\epsilon_2$ and $\epsilon_1$, respectively. 
Applying \cite[Theorem 20]{Polyanskiy} to (\ref{eq:threshold_GBC_DPC2}) and (\ref{eq:threshold_GBC_DPC1}), we obtain the upper bound proposed in Theorem \ref{theo:DPC} which is composed of  three different error events.
(\ref{eq:user2_DPC}) and  (\ref{eq:user1_DPC}) respectively are the total error probability of reconstructing $m_2$ and $m_1$ that are the sum of mis-detection and confusion probabilities. Lastly, we have the probability of the non-permissible inputs $\Pr[X^n \notin F]$ given by (\ref{eq:ub_epsilon_22}) due to the definition of the input signal (\ref{eq:DPC_input}).
\end{proof}

Hereafter, we derive the probability of each error event per message given in the proposed bound of Theorem \ref{theo:DPC} in detail.


\subsubsection{Decoding message 1}
We begin with the evaluation of the mis-detection error probability in decoding $m_1$ using the observation of $Y_1^n$, namely 
\begin{equation}
\Pr \left[ i(U_1^n;Y_1^n) < \gamma'_1 \right]
\end{equation}where $\gamma'_1=I(U_1^n;U_2^n)+\gamma_1$.

The information density $i(U_1^n;Y_1^n) $ is expanded as follows
\begin{align} \label{eq:inf_dens}
&\frac{1}{n} i(U_1^n;Y_1^n) \overset{(a)}{=} \frac{1}{n}\left[h(U_1^n)+h(Y_1^n)-h(U_1^n,Y_1^n) \right]\nonumber \\
&=\frac{1}{n} \left[\log P_{U_1^n,Y_1^n}(U_1^n,Y_1^n)-\log P_{U_1^n}(U_1^n)-\log P_{Y_1^n}(Y_1^n) \right]\nonumber \\
&=\frac{1}{n} \log \frac{ P_{U_1^n,Y_1^n}(U_1^n,Y_1^n)}{P_{U_1^n}(U_1^n)P_{Y_1^n}(Y_1^n)}
\end{align} 
where in step (a), we use the entropy density function defined as $h(x)=-\log P_x(x)$. As for the joint entropy density $h(U_1^n,Y_1^n)$, we use the general correlated bivariate Gaussian distribution as given on the top of the current page by (\ref{eq:inf_dens1}).
\begin{figure*}[!t]
\normalsize
\setcounter{MYtempeqncnt}{\value{equation}}
\setcounter{equation}{24}
\begin{IEEEeqnarray}{rCl}
\label{eq:inf_dens1}
h(U_1^n,Y_1^n)&=& -\log \left[\left[(2\pi)^2 (1-\rho^2)\sigma^2_{U_1}\sigma^2_{Y_1}\right]^{-n/2}\exp \left \{ -\frac{1}{2(1-\rho^2)}\left[\frac{\|U^n_1\|^2}{\sigma^2_{U_1}}+\frac{\|Y^n_1\|^2}{\sigma^2_{Y_1}}-\frac{2\rho \langle U_1^n,Y_1^n\rangle}{\sigma_{Y_1} \sigma_{U_1}}\right]\right\}\right] \nonumber \\
&\overset{(b)}{=}&\frac{n}{2}\log \frac{\alpha P+N_1}{(2\pi)^2N_1 (\alpha P+N_1)}+\frac{1}{n}\log \mathrm{e} \times\left \{ -\frac{1}{2(1-\rho^2)}\left[\frac{\|U^n_1\|^2}{\sigma^2_{U_1}}+\frac{\|Y^n_1\|^2}{\sigma^2_{Y_1}}-\frac{2\rho \langle U_1^n,Y_1^n\rangle}{\sigma_{Y_1} \sigma_{U_1}}\right]\right\}
\end{IEEEeqnarray}
\setcounter{equation}{\value{MYtempeqncnt}}
\hrulefill 
\vspace*{4pt}
\end{figure*}
Here $\rho$ denotes the correlation coefficient defined by $\rho=\frac{\mathrm{Cov}(U_1,Y_1)}{\sigma_{U_1}\sigma_{Y_1}}$ with the covariance and variance terms $\mathrm{Cov}(U_1,Y_1)=\frac{\alpha P (P+N_1)}{\alpha P+N_1}$ and $\sigma^2_{U_1}=\frac{\alpha P (\alpha P\bar{\alpha}P+(\alpha P+N_1)^2)}{(\alpha P+N_1)^2}$ and $\sigma^2_{Y_1}=P+N_1$. The marginal entropy densities of $h(U_1^n)$ and $h(Y_1^n)$ are given below.
\addtocounter{equation}{1}
\begin{align}
h(U_1^n)&=-\log \left[(2 \pi \sigma^2_{U_1})^{-n/2} \exp \left \{-\frac{\|U_1^n\|^2}{2\sigma^2_{U_1}}\right \}\right]  \label{eq:inf_dens2}\\
h(Y_1^n)&=-\log \left[(2 \pi \sigma^2_{Y_1})^{-n/2} \exp \left \{-\frac{\|Y_1^n\|^2}{2\sigma^2_{Y_1}}\right \}\right]  \label{eq:inf_dens3}
\end{align} 
The mutual information between $U_1^n$ and $U_2^n$ using the differential entropies yields
\begin{align}\label{eq:mut_inf_U1U2}
I(U_1^n;U_2^n)&=\frac{1}{2} \log (2\pi \mathrm{e}\sigma^2_{U_1|U_2})^{n}-\frac{1}{2} \log (2\pi \mathrm{e}\sigma^2_{U_1})^{n} \nonumber \\
&=\frac{n}{2} \log \frac{\sigma^2_{U_1|U_2}}{\sigma^2_{U_1}} 
\end{align}
Combining (\ref{eq:inf_dens1}), (\ref{eq:inf_dens2}) with (\ref{eq:inf_dens3}) and substituting into (\ref{eq:inf_dens}),  we obtain $\frac{1}{n}i(U_1^n;Y_1^n)$. Together with (\ref{eq:mut_inf_U1U2}), we finally get the overall expression
\begin{equation}\label{eq:misdet_user1}
\frac{1}{n} \left[i(U_1^n;Y_1^n)-I(U_1^n;U_2^n)\right]=C\left(\frac{\alpha P}{N_1}\right)+v_{1,\mathrm{DP}}
\end{equation} 
The dispersion term denoted by $v_{1,\mathrm{DP}}$ is expanded out as 
\begin{align}
&v_{1,\mathrm{DP}}=\frac{1}{n} \log \mathrm{e} \times \notag \\
&\left[\left(1-\frac{1}{1-\rho^2}\right)\left(\frac{\|U_1^n\|^2}{2\sigma^2_{U_1}}+\frac{\|Y_1^n\|^2}{2\sigma^2_{Y_1}}\right)+\frac{\rho \langle U_1^n,Y_1^n\rangle}{(1-\rho^2)\sigma_{U_1}\sigma_{Y_1}}\right]\nonumber \\
&\overset{(b)}{=}\frac{1}{n} \log \mathrm{e} \times\left[\|Y_1^n\|^2 a_1+ \|U_1^n\|^2 a_2+\langle U_1^n,Y_1^n\rangle a_3 \right]\nonumber \\
&=\frac{1}{n} \log \mathrm{e} \sum _{i=1}^n [X_{1,i} X_{2,i} Z_{1,i}] \mathbf{A}_{\mathrm{DP}}[X_{1,i} X_{2,i} Z_{1,i}] ^T \nonumber \\
&=\frac{1}{n} \log \mathrm{e} \sum _{i=1}^n [\mu_{1,i} \mu_{2,i} \mu_{3,i}] \mathbf{P}_{1,\mathrm{DP}}\mathbf{A}_{\mathrm{DP}} \mathbf{P}_{1,\mathrm{DP}}[\mu_{1,i} \mu_{2,i} \mu_{3,i}] ^T 
\end{align} where $\mu_{k,i} \sim N(0,1)$ for $k=1,2,3$. In step (b) we defined the following variables $a_1=\frac{1}{2\sigma^2_{Y_1}}(1-\frac{1}{1-\rho^2})$, $a_2=\frac{1}{2\sigma^2_{U_1}}(1-\frac{1}{1-\rho^2})$ and $a_3=\frac{\rho}{(1-\rho^2)\sigma_{U_1}\sigma_{Y_1}}=1/N_1$ with $\mathbf{P}_{1,\mathrm{DP}}=\mathrm{diag}(\sqrt{\alpha P},\sqrt{\bar{\alpha} P},\sqrt{N_1})$. $\mathbf{A}_{\mathrm{DP}}$ is a 3-by-3 matrix and given on the top of the next page by (\ref{matrix_A}).
\begin{figure*}[!t]
\normalsize
\setcounter{MYtempeqncnt}{\value{equation}}
\setcounter{equation}{30}
\begin{IEEEeqnarray}{rCl}
 \label{matrix_A}
\mathbf{A}&=\begin{bmatrix}
(a_1+a_2+ a_3)& (a_1 b_2+a_2 +\frac{a_3}{2}(b_2+1))&(a_2+a_3/2)\\
 (a_1 b_2+a_2 +\frac{a_3}{2}(b_2+1))& (a_1 b_2^2+a_2+a_3 b_2)&(a_2+\frac{a_3 b_2}{2})\\
(a_2+a_3/2)&(a_2+\frac{a_3 b_2}{2})&a_2
\end{bmatrix}, 
\end{IEEEeqnarray}
\setcounter{equation}{\value{MYtempeqncnt}}
\hrulefill 
\vspace*{4pt}
\end{figure*}
The matrix product $\mathbf{P}_{1,\mathrm{DP}}\mathbf{A}_{\mathrm{DP}} \mathbf{P}_{1,\mathrm{DP}}$ has only two non-zero eigenvalues $\lambda_{1,\mathrm{DP}}=\frac{1}{2}\frac{\alpha P(P+N_1)}{N_1^2+2\alpha P N_1+\alpha P^2}$ and $\lambda_{2,\mathrm{DP}}=-\lambda_{1,\mathrm{DP}}$. The dispersion term which follows a special case of the gamma distribution becomes
\addtocounter{equation}{1}\begin{equation}\label{eq:v_1DP}
v_{1,\mathrm{DP}}=\frac{1}{n} \log \mathrm{e}  \times \lambda_{1,\mathrm{DP}}\sum _{i=1}^n \left(\chi^2_{1,i}- \chi_{2,i}^2\right)
\end{equation} for $\chi^2_{j,i}$, $j=1,2$ are chi-squared random variables with one degree-of-freedom.
Substituting (\ref{eq:v_1DP}) into (\ref{eq:misdet_user1}), the probability of mis-detecting $m_1$ on user 1 is evaluated through the c.d.f. of the channel dispersion term $v_{1,\mathrm{DP}}$ as
\begin{equation} \label{eq:DPC_user1_misdet}
\Pr \left[ i(U_1^n;Y_1^n) < \gamma'_1 \right] 
\leq \Pr \left[v_{1,\mathrm{DP}} < \frac{\gamma'_1/n -C \left(\frac{\alpha P}{N_1 }\right)}{\log \mathrm{e}}   \right].
\end{equation} 
The second term in (\ref{eq:user1_DPC}) is the confusion error probability in decoding $m_1$ using the observation of $Y_1^n$ is expanded out as in (\ref{sec_term}) on the next page
\begin{figure*}[!t]
\normalsize
\setcounter{MYtempeqncnt}{\value{equation}}
\setcounter{equation}{33}
\begin{IEEEeqnarray}{rCl}
\label{sec_term}
\frac{\left[i(U_1^n;\bar{Y}_1^n|X^n_1) -I(U_1^n;U_2^n) \right]}{n}&&=\frac{1}{n}\log \frac{P_{U_1^n,Y_1^n|X^n_1=C_k^n}}{P_{U^n_1}P_{Y^n_1}}-\frac{1}{n}I(U_1^n;U_2^n)\nonumber \\
&&\overset{(a)}{=}C\left(\frac{\alpha P}{N_1}\right)+\frac{\log \mathrm{e}}{n} \left\{a_3 \langle b_2 X_2^n+C_k^n-C_l^n,C_k^n-C_l^n+X^n_2+Z^n_1\rangle +a_4\|b_2 X_2^n+C_k^n-C_l^n\|^2\right. \nonumber \\
&&\left.+a_5 \|C_k^n-C_l^n+X^n_2+Z^n_1\|^2+a_6\|b_2 X_2^n+C_k^n\|^2+a_7 \|C_k^n+X^n_2+Z^n_1\|^2\right\}
\end{IEEEeqnarray}
\setcounter{equation}{\value{MYtempeqncnt}}
\hrulefill 
\vspace*{4pt}
\end{figure*}
with $a_4=-\frac{1}{[2(1-\rho^2)\sigma^2_{U_1}]}$, $a_5=-\frac{1}{[2(1-\rho^2)\sigma^2_{Y_1}]}$, $a_6=1/(2\sigma^2_{U_1})$, $a_7=1/(2\sigma^2_{Y_1})$. 
The confusion error event in (a) of (\ref{sec_term}) is rewritten as 
\addtocounter{equation}{1}
\begin{equation}
\frac{1}{n}\left[i(U_1^n;\bar{Y}_1^n|X^n_1=C_k^n) -I(U_1^n;U_2^n) \right]=C\left(\frac{\alpha P}{N_1}\right)+v_{c,\mathrm{DP}}
\end{equation} 
where the dispersion term $v_{c,\mathrm{DP}}$ is derived by
\begin{align}
v_{c,\mathrm{DP}}&=\frac{1}{n}\log \mathrm{e} \sum _{i=1}^{n}[X_{2,i} C_{k,i} C_{l,i} Z_{1,i}] \mathbf{D}_c [X_{2,i} C_{k,i} C_{l,i} Z_{1,i}]^T \nonumber \\
&=\frac{1}{n}\log \mathrm{e} \sum _{i=1}^{n}[\mu_{1,i} \mu_{2,i} \mu_{3,i} \mu_{4,i}] \mathbf{P}_c\mathbf{D}_c \mathbf{P}_c [\mu_{1,i} \mu_{2,i} \mu_{3,i} \mu_{4,i}]^T \nonumber \\
&=\frac{1}{n}\log \mathrm{e} \sum _{i=1}^{n}\sum_{t=1}^4 [\lambda_{DP,t}^2 \chi_{t,i}^2] \label{eq:36}
\end{align} $\mathbf{P}_c$ is a $4 \times 4$ diagonal matrix that is $\mathbf{P}_c=\mathrm{diag}(\sqrt{\bar{\alpha} P},\sqrt{\alpha P},\sqrt{\alpha P},\sqrt{N_1})$ and $\mathbf{D}_c$ represents the factors of $X_{2,i}$, $C_{k,i}$, $C_{l,i}$, $Z_{1,i}$ from (\ref{sec_term}) in a form of $4\times 4$ matrix. In (\ref{eq:36}), the eigenvalues of the matrix product $\mathbf{P}_c\mathbf{D}_c \mathbf{P}_c$ denoted by $\lambda_{DP,t}$ are 9th order polynomials of the SNR which are omitted here. $v_{c,\mathrm{DP}}$ is composed of a weighted sum of chi-squared variables with $n$ degrees of freedom (with different positive and negative weighs) that does not follow a known probability distribution in the literature.
The confusion error probability in decoding message 1 becomes
\begin{equation} \label{eq:conf_DPC_user1}
\Pr\left[i(U_1^n;\bar{Y}_1^n)  > \gamma'_1 \right]=
\Pr \left[v_{c,\mathrm{DP}}> \frac{\frac{1}{n}\gamma'_1-C\left(\frac{\alpha P}{N_1}\right)}{\log \mathrm{e}} \right]
\end{equation}  Combining (\ref{eq:conf_DPC_user1}) with (\ref{eq:DPC_user1_misdet}), we obtain the total error probability in reconstructing message 1 as follows
\begin{equation}\label{eq:epsilon_1}
\epsilon_1\leq\Pr \left[v_{1,\mathrm{DP}}> \zeta_1 \right]+\Pr \left[v_{c,\mathrm{DP}}> \zeta_1 \right]
\end{equation} for $\zeta_1=\frac{\gamma'_1-n C\left(\frac{\alpha P}{N_1}\right)}{n \log \mathrm{e}}$.
\subsubsection{Decoding message 2}
For the probability of error in decoding $m_2$ using the observation of user 2 $Y_2^n$, we derive the information density $i(U_2^n;Y_2^n)$ associated with $R_2$ given by (\ref{eq:R2}) in the asymptotic capacity. Due to the decoding rule (\ref{eq:threshold_GBC_DPC2}) with $U^n_2=X^n_2$, the sum probability in (\ref{eq:user2_DPC})
is equal to the one of (\ref{eq:DT_SP_2}) in the setting of superposition coding \cite[Section III.1]{UG_2017}. For a detailed derivation, the reader is referred to \cite[Section III.1]{UG_2017}.
The outage/misdetection error probability in decoding $m_2$ is obtained as
\begin{equation} \label{eq:misdet_2_final}
\Pr \left[ i(U_2^n;Y_2^n) < \gamma_2\right]
=\Pr \left[v_{22} < \zeta_2  \right]
\end{equation} with $\zeta_2=\frac{\gamma_2-n C\left(\frac{\bar{\alpha}P}{\alpha P+N_2}\right)}{\log \mathrm{e}}$. $v_{22}$ is a weighted difference of chi-squared variables with $n$ degrees of freedom 
\begin{equation}\label{v_11}
v_{22}=\frac{1}{n}\log \mathrm{e}\times \lambda_1\sum_{i=1}^n [\chi_{1,i}^2- \chi_{2,i}^2]
\end{equation} for $\lambda_1=\frac{\sqrt{\bar{\alpha}P}}{2\sqrt{(P+N_2)}}$ and $\chi_{j,i}^2$ denotes a chi-squared random variable with one degree of freedom. The weighted difference of two chi-squared variables with $n$ degrees of freedom corresponds to the variance-gamma distribution\footnote{{The variance-gamma distribution with $\lambda>0$ degrees of freedom, the real-valued location parameter $m$ and the asymmetry parameter $b$,  is defined as
\begin{equation}
f(t)=\frac{\theta^{2\lambda}|t-m|^{\lambda-1/2}K_{\lambda-1/2}(\delta|t-m|)}{\sqrt{\pi}\Gamma(\lambda)(2\delta)^{\lambda-1/2}}\mathrm{e}^{b(t-m)},\;\;-\infty <t<\infty \nonumber 
\end{equation} where $K_{\lambda}$ and $\Gamma$ denote the modified Bessel function of the second kind and the gamma function, respectively for real $\delta$, $\theta=\sqrt{\delta^2-b^2}>0$.}}  with the location parameter $m$ and asymmetry parameter denoted by $b$ are both 0 and $\delta=\theta=\sqrt{\frac{P+N_2}{\bar{\alpha} P}}$.

As for the confusion error probability in decoding $m_2$ on user 2, we have $\frac{1}{n} i(U_2^n;\bar{Y}_2^n)=C \left(\frac{\bar{\alpha}P}{(\alpha P+N_2)}\right)+v_{c,2}$ where the additional term $v_{c,2}$ is
\begin{equation} \label{eq:vc2}
v_{c,2}=\frac{1}{n}\log \mathrm{e} \sum_{i=1}^{n} (\lambda_{c1}\chi_{1,i}^2+\lambda_{c2}\chi_{2,i}^2)
\end{equation}  with ${\lambda_{c1}=\frac{-\left(\bar{\alpha}P+\sqrt{\bar{\alpha}P(P+N_2)}\right)}{2(\alpha P+N_2)}}$ and ${\lambda_{c2}=\frac{-\bar{\alpha}P+\sqrt{\bar{\alpha}P(P+N_2)}}{2(\alpha P+N_2)}}$.
Consequently, the evaluation of the bound on the total probability of error in decoding the second message (\ref{eq:user2_DPC}) follows the same way as $\epsilon_2$ in \cite[Section III.1]{UG_2017} given by
\begin{equation}\label{eq:epsilon2}
\epsilon_2\leq\Pr \left[v_{22} < \zeta_2  \right]+\eta_2 \Pr\left[v_{c,2}>\zeta_2\right]
\end{equation}

Due to the average power constraint on the codebooks, the probability of the input being confined to the feasible yields $\Pr\{X^n \notin F\}=0$. Combining (\ref{eq:epsilon_1}) with (\ref{eq:epsilon2}), the derivation of dependence testing bound in Theorem \ref{theo:DPC} is completed.

\section{$\kappa \beta$ Bound --Writing on Dirty Paper \label{sec:kb_DPC}}

\subsection {Peak power \label{subsec:model_kb}}

For the $\kappa \beta$ bound to be introduced in this part, we assume codebooks constructed as in (\ref{eq:DPC_input}) subject to the following power constraints on the encoded messages
\begin{equation} \label{eq:Power_X}
\!\!\!\|x^n (m_1,m_2)\|^2 = \sum_{i=1}^{n} |x_{1,i}(m_1,m_2)+x_{2,i}(m_2)|^2 \leq n P,
\end{equation} and $\forall m_1, m_2$ where we assume equal power per codeword for $X_{2,i}(m_2)$, $i=1,\cdots,n, \; m_2=1,\cdots,M_2$ and $m_1=1,\cdots,M_1$ s.t.
\begin{equation} 
\label{eq:Power_X_2}
\sum_{i=1}^{n} |x_{2,i}(m_2)|^2=n P_2, \; \forall m_2.
\end{equation}
For $X_{1,i}(m_1,m_2)$, we have 
\begin{align} \label{eq:constraint_derivation}
\|x^n (m_1,m_2)\|^2& 
= n P_1(m_1,m_2) + nP_2 \nonumber \\
&+2 n\rho(m_1,m_2) \sqrt{P_1 (m_1,m_2)P_2} \nonumber \\
&\leq nP
\end{align} with $\rho(m_1,m_2)\in [-1,1]$ defined as in the following
\begin{equation} \label{eq:Power_X_1}
\sum_{i=1}^{n} x_{1,i}^2(m_1,m_2)=n P_1(m_1, m_2), \; \forall m_1, m_2.
\end{equation} 

According to (\ref{eq:constraint_derivation}), if $X^n_{1}(m_1,m_2)$ is chosen in the null space of $X^n_{2}(m_2)$ for each $m_2$ then a constant power can be assigned for $X_{1,i}(m_1,m_2)$ as $P_1(m_1,m_2)=P-P_2$. 
This model is considered for the $\kappa\beta$ bound covered in Section \ref{sec:beta}.
Due to (\ref{eq:constraint_derivation}), we define the feasible set denoted by $F_n$ as
\begin{equation}\label{eq:feasible_kb}
F_n \dfn \{x^n : \|x^n \|^2 \leq nP \} \subset \mathbb{R}^n
\end{equation} 
In \cite[Theorem IV]{UG_2017}, the authors presented the multiuser adaptation of \cite[Theorem 41]{Polyanskiy} for the two-receiver Gaussian BC in the setting of SPC. In order to apply dirty paper coding to obtain the two-user equivalent of the $\kappa \beta$ bound in this setting, we derive the two mutual information functions from Section \ref{sec:DT_DPC} based on the distributions of the output signals and the conditional distributions of output given input signal. Note that the input distribution is not limited to any type of distribution as in the previous part.

\subsection{Derivation of $\beta$ \label{sec:beta}}
For the second user, $Y_2^n$, we define $P_{Y_2^n}\sim \mathcal{N}(0,(P+N_2) \mathbf{I}_n)$.
Using the definition of (\ref{eq:inf_density}), $i(U_2^n;Y_2^n)$ under $P_{Y_2^n}$ is given as 
\begin{align} \label{eq:Gn2}
&G_{n_2}=\frac{n}{2} \log \sigma^2_{Y_2}-\frac{nP_2}{P_1 +N_2}\frac{\log \mathrm{e}}{2} \nonumber \\
&+\frac{1}{2}\log \mathrm{e}\sum_{i=1}^{n}  \left[(1-\sigma^2_{Y_2})S_i^2+2\sqrt{P_2/(P_1 +N_2)}\sigma_{Y_2}S_i\right] 
\end{align} where $S_i \sim \mathcal{N}(0,1)$ for $i=1,\cdots,n$. Under the conditional distribution $P_{Y_2^n|U_2^n}$ the same information density  yields
\begin{align}\label{eq:Hn2}
&H_{n_2}=\frac{n}{2} \log  (\sigma^2_{Y_2}) +\frac{nP_2}{P_1 +N_2}\frac{\log \mathrm{e}}{2\sigma^2_{Y_2}} \nonumber \\
&+\frac{\log \mathrm{e}}{2\sigma^2_{Y_2}} \sum_{i=1}^{n} \left[\left(1-\sigma^2_{Y_2}\right)S_i^2+2\sqrt{P_2/(P_1 +N_2)}S_i\right]
\end{align} Note that, in decoding $m_2$, the corresponding mutual information functions, hence the associated information densities, are the same in both cases of superposition and dirty paper coding since $U_2^n=X_2^n$. 
As for user 1, we set $P_{Y_1^n}\sim \mathcal{N}(0,(P+N_1) \mathbf{I}_n)$. In a similar manner to user 2, we denote $i(U_1^n;Y_1^n)$ by $G_{n_1}$ based on $P_{Y_1^n}$ which is given by
\begin{align}\label{eq:Gn1}
&G_{n_1}=\frac{n}{2} \log \sigma^2_{Y_1}-\frac{nP}{N_1}\frac{\log \mathrm{e}}{2} \nonumber \\
&+\frac{1}{2}\log \mathrm{e}\sum_{i=1}^{n}  \left\{(1-\sigma^2_{Y_1})S_i^2+2 \sqrt{P/N_1}\sigma_{Y_1}S_i \right\} 
\end{align}
Lastly, based on the conditional probability distribution $P_{Y_1^n|U_1^n}$ for user 1, the corresponding information density function is denoted by $H_{n_1}$ where
$Y_1^n|U_1^n$ yields $b_1 X_2^n+Z_1^n$ with $b_1=(1-b_2)=\frac{N_1}{P_1+N_1}$ due to the input signal defined by (\ref{eq:DPC_input}). $H_{n_1}$ by the following equality.
\begin{align} \label{eq:Hn1}
&H_{n_1}=\frac{n}{2} \log  (\sigma^2_{Y_1}) +\frac{nb_1 P_2}{N_1}\frac{\log \mathrm{e}}{2\sigma^2_{Y_1}} \nonumber \\
&+\frac{\log \mathrm{e}}{2\sigma^2_{Y_1}} \sum_{i=1}^{n} \left[\left(1-\sigma^2_{Y_1}\right)S_i^2+2\sqrt{b_1 P_2/N_1}S_i\right]
\end{align}

In a similar fashion to \cite[Theorem 40]{Polyanskiy} for the degraded Gaussian BC, $\beta_{\epsilon_j,j}$, for $j=1,2$ is 
\begin{align} \label{eq:beta_GBC}
\beta_{\epsilon_2,2}&=\Pr[G_{n_2} \geq \gamma_2] \\
&=\Pr \left[\sum_{i=1}^n \left(S_i+\sqrt{\frac{P_2 \sigma_{Y_2}^2}{P_1+N_2}}\frac{1}{1-\sigma_{Y_2}^2}\right)^2\geq \frac{2\gamma_2}{\log \mathrm{e}(1-\sigma_{Y_2}^2)} \right. \\
&\left.-\frac{n\log \sigma_{Y_2}^2}{(1-\sigma_{Y_2}^2)\log \mathrm{e}}+\frac{nP_2}{(P_1+N_2)(1-\sigma_{Y_2}^2)}\left(1+\frac{\sigma_{Y_2}^2}{1-\sigma_{Y_2}^2}\right)\right]
\end{align}
\begin{align}
\beta_{\epsilon_1,1}&=\Pr[G_{n_1} \geq \gamma_1]\\
&=\Pr\left[\sum_{i=1}^n \left(S_i+\frac{\sqrt{P /N_1}\sigma_{Y_1}}{1-\sigma_{Y_1}^2}\right)^2\geq \frac{2\gamma_1}{(1-\sigma_{Y_1}^2)\log \mathrm{e}} \right. \\
&\left.-\frac{n \log \sigma_{Y_1}^2}{\log \mathrm{e}(1-\sigma_{Y_1}^2)} +\frac{nP}{N_1(1-\sigma_{Y_1}^2)}\left(1+\frac{\sigma_{Y_1}^2}{1-\sigma_{Y_1}^2}\right)\right]
\end{align} where 
\begin{align}
&\Pr[H_{n_2} \geq \gamma_2]=1-\epsilon_2 \\
&=\Pr\left[\sum_{i=1}^n \left(S_i+\frac{\sqrt{P_2/(P_1+N_2)}}{1-\sigma_{Y_2}^2}\right)^2 \geq \frac{2\gamma_2\sigma_{Y_2}^2}{\log \mathrm{e}(1-\sigma_{Y_2}^2)}\right.\label{eq:59} \\
&\left.   -\frac{n \sigma_{Y_2}^2\log \sigma_{Y_2}^2}{\log \mathrm{e} (1-\sigma_{Y_2}^2)}-\frac{n P_2/(P_1+N_2)}{(1-\sigma_{Y_2}^2)}\left(1-\frac{1}{1-\sigma_{Y_2}^2}\right)   \right]\label{eq:60}\\
&=Q_{n/2}\left(\sqrt{\lambda_2},\sqrt{t_2}\right)
\end{align} where $\lambda_2$ is the non-centrality parameter $n\frac{P_2/(P_1+N_2)}{(1-\sigma_{Y_2}^2)^2}$ and $t_2=\frac{2\gamma_2\sigma_{Y_2}^2}{\log \mathrm{e}(1-\sigma_{Y_2}^2)} -\frac{n \sigma_{Y_2}^2\log \sigma_{Y_2}^2}{\log \mathrm{e} (1-\sigma_{Y_2}^2)}-\frac{n P_2}{(P_1+N_2)(1-\sigma_{Y_2}^2)}\left(1-\frac{1}{1-\sigma_{Y_2}^2}\right) $ that is (\ref{eq:60}).
\begin{align}
&\Pr[H_{n_1} \geq \gamma_1]=1-\epsilon_1\\
&=\Pr \left[\sum_{i=1}^n \left(S_i+\frac{1}{1-\sigma_{Y_1}^2}\sqrt{\frac{b_1 P_2}{N_1}}\right)^2 \geq \frac{2\sigma_{Y_1}^2\gamma_1}{\log \mathrm{e} (1-\sigma_{Y_1}^2)} \right. \\
&\left. -\frac{n\sigma_{Y_1}^2 \log \sigma_{Y_1}^2}{\log \mathrm{e}(1-\sigma_{Y_1}^2)}-\frac{n b_1 P_2}{N_1(1-\sigma_{Y_1}^2)}\left(1-\frac{1}{1-\sigma_{Y_1}^2}\right)\right] \\
&=Q_{n/2}\left(\sqrt{\lambda_1},\sqrt{t_1}\right)
\end{align} where $\lambda_1=n\frac{P /N_1}{(1-\sigma_{Y_1}^2)^2}$ and $t_1=\frac{2\sigma_{Y_1}^2\gamma_1}{\log \mathrm{e} (1-\sigma_{Y_1}^2)}-\frac{\sigma_{Y_1}^2 \log \sigma_{Y_1}^2}{\log \mathrm{e}(1-\sigma_{Y_1}^2)}-\frac{n b_1 P_2 \sigma_{Y_1}^2}{N_1(1-\sigma_{Y_1}^2)^2}$ for $\gamma_2>0$ and $\gamma_1 > \log \sqrt{\frac{\sigma^2_{U_1|U_2}}{\sigma^2_{U_1}}}$. In the next theorem, we state the obtained achievability bound.
\begin{theorem} \label{theo:kappa-beta-DPC}
The following lower bounds on the maximal code sizes $M^*_j$ in a two receiver Gaussian BC hold for any $\epsilon_j$,  $n\geq 1$, $\tau_j\in [0,1]$ and the chosen $P_{Y_j}$ for $j=1,2$ with $F_n$ as defined by (\ref{eq:feasible_kb})
\begin{align}
&M^*_1 \geq \frac{\kappa_{\tau_1,1}(F_n,P_{Y^n_1})}{\beta_{1-\epsilon_1+\tau_1,1}(x,P_{Y^n_1})} \\
&M^*_2 \geq \frac{\kappa_{\tau_2,2}(F_n,P_{Y^n_2})}{\beta_{1-\epsilon_2+\tau_2,2}(x,P_{Y^n_2})}
\end{align} where  
\begin{equation}
\kappa_{\tau_j,j}(F_n,P_{Y^n_j})=P_{0,j}\left[\frac{p_{1,j}(r)}{p_{0,j}(r)} \geq \psi_j\right] 
\end{equation} 
with $\psi_j$ satisfying
\begin{equation}
P_{1,j}\left [\frac{p_{1,j}(r)}{p_{0,j}(r)} \geq \psi_j \right]=\tau_j.
\end{equation}
The probability distributions $p_{0,j}(r)$ and $p_{1,j}(r)$  are defined as
$p_{0,j}(r)=\frac{1}{\Gamma(n/2) \omega_j^{n/2}}r^{n/2-1}\mathrm{e}^{r/\omega_j}$ and 
${p_{1,j}(r)=\frac{1}{2}\mathrm{e}^{(r+\upsilon_j)/2}(\frac{r}{\upsilon_j})^{n/4-1/2}I_{n/2-1}(\sqrt{\upsilon_j r})}$ 
with the modified Bessel function of the first kind ${I_{b}(y)=(y/2)^b\sum_{j=0}^{\infty} \frac{(y^2/4)^l}{l! \Gamma(b+l+1)}}$
and the following parameters for $k=1$, $\omega_1=2(N_1+P)$ and $\upsilon_1=\frac{n b_1 P_2}{N_1}$, for $k=2$, $\omega_2=2(N_2+P)$ and $\upsilon_2=\frac{nP_2}{P_1 +N_2}$.
\end{theorem}
\begin{proof}
The reader is referred to the proofs of \cite[Theorems 25, 40 and 42]{Polyanskiy} for a more detailed proof on the single user bounds. $\beta_{1-\epsilon_j+\tau_j,j}(x,P_{Y^n_j})$ for $a=1-\epsilon_j+\tau_j$ given by (\ref{eq:beta_GBC}) is derived using the information densities $G_{n_2}$, $H_{n_2}$, $G_{n_1}$ and $H_{n_1}$ respectively defined by (\ref{eq:Gn2})-(\ref{eq:Hn1}). For evaluating $\kappa_{\tau_j,j}$, we set the following central and non-central chi-squared distributions $P_{0,j}$ and $P_{1,j}$, respectively using dirty paper coding as 
\begin{align}
P_{0,1}&\sim \sum_{i=1}^{n}(P+N_1)S_i^2, \\
P_{1,1}&\sim \sum_{i=1}^{n}\left(\sqrt{b_1 P_2 }+\sqrt{N_1}S_i \right)^2 \\
P_{0,2}&\sim \sum_{i=1}^{n}( P +N_2) S_i^2, \\
P_{1,2}&\sim \sum_{i=1}^{n}\left(\sqrt{P_2} +\sqrt{(P_1+N_2)}S_i \right)^2 
\end{align}
\end{proof}
\section{Discussion and Conclusion \label{sec:discuss}}
We presented the dependence testing and $\kappa \beta$ bounds for the two-receiver Gaussian BC in the setting of dirty paper coding with finite blocklength. Despite the well-known optimality of two coding schemes, superposition and dirty paper coding in the asymptotic regime, evidently dirty paper coding outperforms superposition coding in the non-asymptotic regime for a larger system with many users, $k\geq 2$.
The main difference between the two coding schemes in the non-asymptotic regime is the probability of error in decoding $m_1$ using the observation of $Y_1^n$, does not depend on the correct decoding of $m_2$ in the case of dirty paper coding as shown in Theorem \ref{theo:DPC} as opposed to the superposition coding of \cite[Theorem 2]{UG_2017}.
\bibliography{DT_biblio}
\bibliographystyle{IEEEtran}
\end{document}